\documentclass[aps,prxquantum,twocolumn,superscriptaddress,longbibliography]{revtex4-2}

\pdfoutput=1

\usepackage{amsmath,amssymb,amsfonts}
\usepackage{amsthm}
\usepackage{physics}
\usepackage{braket}
\usepackage{microtype}

\usepackage{graphicx}
\usepackage{tikz}
\usepackage{subcaption}

\usepackage{multirow}
\usepackage{booktabs}

\usepackage[ruled,vlined]{algorithm2e}

\usepackage{xcolor}
\usepackage{tcolorbox}

\usepackage[english]{babel}

\usepackage[colorlinks=true,linkcolor=blue,citecolor=blue,urlcolor=blue]{hyperref}

\newtheorem{theorem}{Theorem}

\newtheorem{corollary}[theorem]{Corollary}
\newtheorem{lemma}[theorem]{Lemma}

\newcommand{\Var}{\operatorname{Var}}

\begin{document}

\title{Certified Fidelity Susceptibility from Classical Shadows}

\author{Kiran Adhikari}
\email{kiran.adhikari@tum.de}
\affiliation{Emmy Noether Group for Theoretical Quantum Systems Design\\
                    Technical University of Munich, Germany}

\begin{abstract}
Fidelity susceptibility is a useful probe of quantum phase transitions and quantum metrology, but its direct evaluation on a quantum device is challenging. By integrating Krylov-subspace methods with a resolvent-based reformulation, we present a method for certifying fidelity susceptibility using randomized single-copy measurements across repeated ground-state preparations. The resulting approximations form monotone lower bounds to the exact fidelity susceptibility and converge geometrically. We also present applications to metrology and linear static susceptibilities. 

\end{abstract}

\maketitle

\section{Introduction}

 Quantum Phase transitions \cite{Sachdev_1999} are crucial for diverse phenomena ranging from superconductivity \cite{Ginzburg:1950sr} to topological phases \cite{Wen:1989zg, kharel2026pseudoentropytopologicalphases}. Fidelity susceptibility provides a powerful tool for detecting quantum phase transitions even without prior knowledge of the order parameters \cite{GU_2010, Gu_2008, Zanardi_2006, You_2007, Wang_2015}. Moreover, it is also related to Fisher Information \cite{Liu_2019}, and thus to quantum metrology \cite{Giovannetti_2006,PhysRevLett.100.100503, Hyllus_2012, Pezz__2009}. 
 
Estimating fidelity susceptibility directly on a quantum device is challenging as the standard representation requires Hamiltonians' full spectrum \cite{Wang_2015}, while an alternative finite-difference approach requires fidelities between nearby ground states, which amplify statistical noise \cite{Di_Matteo_2022}. Classical shadow was proposed as a powerful technique for estimating many observables from a few measurements \cite{Huang_2020, Huang_2021}, which motivates the question if fidelity susceptibility can be reconstructed directly from measurements of a single reference ground state.

Recent approaches include variational methods \cite{Tan_2021, Patel_2025, Minervini_2026} and the quantum singular value transformation (QSVT) \cite{nkrh-5lbl,Gily_n_2019, PRXQuantum.2.040203}, but either they require deep circuits, or they are not guaranteed to work in practical situations. Upper and lower bounds on the fidelity susceptibility, expressed in terms of thermodynamic susceptibilities and thermal averages, have been derived \cite{PhysRevE.85.031115}, but not in the context of quantum computation.
Related randomized measurement protocols exist that construct lower bounds to the quantum Fisher information \cite{rathQuantumFisherInformation2021,
vitaleRobustEstimationQuantum2024}, while Krylov shadow tomography
\cite{zhangKrylovShadowTomography2025a,
wangSuperiorityKrylovShadow2026,alishahiha2026krylov,chowdhury2026polynomial} provides an efficient moment-based reconstruction of such quantities, though it requires multicopy snapshots \cite{Huang_2020}. Krylov techniques have also been explored in the context of continuous-variable systems \cite{Adhikari:2022whf,Adhikari:2023evu,Adhikari:2023umo} and information scrambling \cite{Adhikari:2023umo}.

In this work, we integrate the resolvent formulation of fidelity
susceptibility with a Krylov-subspace construction tailored to the
Hamiltonian response~\cite{Morawetz_2025}.  The resulting approximation depends only on moments, which are linear expectation values of the reference ground-state density
matrix and can therefore be estimated from a single-copy classical-shadow data
set. We obtain a monotone hierarchy of certified lower bounds to the fidelity
susceptibility. The quantum part of the protocol, therefore, requires only randomized single-copy measurements, while the
Krylov reconstruction and certification are performed classically.

\section{Fidelity Susceptibility in Classical Shadow Friendly Form}
\label{subsec:fidelity_definition}

Physically, Fidelity susceptibility quantifies the ground-state sensitivity to perturbations of a control parameter $\lambda$. Throughout this work, we assume a finite-dimensional Hilbert space, the Hamiltonian $H(\lambda)$ has a non-degenerate ground-state energy $E_0(\lambda)$ with a normalized ground state
$|\Psi_0(\lambda)\rangle$ and spectral gap
$\Delta = E_1 - E_0 > 0$. These are standard assumptions in studies of quantum phase transitions
and ground-state preparation \cite{Lin_2020}.  We then write the Hamiltonian as:
\begin{equation}
    H(\lambda+\epsilon)
=
H(\lambda)+\epsilon H_I+O(\epsilon^2), \quad H_I
=
\partial_\lambda H(\lambda)
\label{eq:perturbed_hamiltonian}
\end{equation}
where, $H_I$ is the driving term, $\epsilon$ is the small perturbation, and $\lambda$ is the control parameter. The ground-state fidelity is defined by
\begin{equation}
   F(\lambda,\lambda+\epsilon)
=
\left|
\langle \Psi_0(\lambda)|\Psi_0(\lambda+\epsilon)\rangle
\right| 
\end{equation}
where, $ \ket{\Psi_0(\lambda)}$ and $\ket{\Psi_0(\lambda+\epsilon)}$ are the unperturbed and perturbed ground states respectively.  For small $\epsilon$, the fidelity becomes
$F(\lambda,\lambda+\epsilon)
=
1-\frac{\epsilon^2}{2}\chi_F(\lambda)+O(\epsilon^3),$ where $\chi_F(\lambda)$ is the fidelity susceptibility:
\begin{equation}
\chi_F(\lambda)
=
-\left.
\frac{\partial^2}{\partial \epsilon^2}
\ln F(\lambda,\lambda+\epsilon)
\right|_{\epsilon=0}.
\label{eq:fidelity_susceptibility_definition}
\end{equation}
There are several reformulations of fidelity susceptibility \cite{Campos_Venuti_2007, GU_2010, You_2007}. A direct perturbative expansion (Appendix \ref{sec:spectral_representation}) yields: 
\begin{equation}
\chi_F(\lambda)
=
\sum_{j \neq 0}
\frac{
|\langle\Psi_j (\lambda)|H_I|\Psi_0 (\lambda)\rangle|^2
}{
[E_j(\lambda) -E_0 (\lambda)]^2
}.
\label{eq:fidelity_susceptibility_spectral_main}
\end{equation}
Although exact, it is not suited for classical shadows, as it involves the excited-state spectrum and transition matrix elements. To eliminate this, we define
\begin{equation}
K:=H-E_0 I, \qquad
    K^+
=
\sum_{j=1}^{D-1}
\frac{1}{(E_j-E_0)}
|\Psi_j\rangle\langle\Psi_j|.\nonumber
\end{equation}
where, $K^+$ is the Moore--Penrose pseudoinverse, and the parameter $\lambda$ is dropped for notational simplicity. Introducing connected perturbation 
\begin{equation}
    B
    :=
    H_I-\langle\Psi_0|H_I|\Psi_0\rangle I, \qquad \ket{\phi} := B \ket{\psi_0},
\end{equation}
the fidelity susceptibility becomes
\begin{equation}
\label{eq:chi_F}
    \chi_F
    =
    \bra{\phi}(K^+)^2\ket{\phi}.
\end{equation}
The remaining difficulty is that $(K^+)^2$ is highly non-local and cannot be readily extracted via quantum devices. 

Instead of implementing it directly, we employ polynomial functional calculus (Appendix \ref{subsec:fs_moment_reduction}) to reduce the problem to response moments
\begin{align}
M_k
&:=
\bra{\phi}K^k\ket{\phi}
\nonumber\\
&=
\bra{\Psi_0}BK^kB\ket{\Psi_0}
=
\operatorname{Tr}\!\left(\rho_0 BK^kB\right),
\label{eq:response_moments}
\end{align}
where $\rho_0=\ket{\Psi_0}\bra{\Psi_0}$. Each $M_k$ is therefore a linear
functional of $\rho_0$ and can be estimated from
a classical-shadow data set. The quantum estimation problem is thus
reduced to the classical moment-reconstruction problem
\begin{equation}
\{M_k\}
\longrightarrow
\chi_F .
\label{eq:classical_moment_problem}
\end{equation}
The next section develops the Krylov reconstruction that performs this
map using finitely many moments.

\section{Krylov reconstruction}
The central idea of this approach is to approximate $K^{+}\ket{\phi}$  rather than the operator $K^{+}$.  For this, we introduce the order-$n$ Krylov subspace \cite{10.1093/acprof:oso/9780199655410.001.0001,Adhikari:2025vdl,Alishahiha:2026fnu,Gautschi_1996}
\begin{equation}
    \mathcal{V}_n
    :=
    \operatorname{span}
    \left\{
        K\ket{\phi},
        K^2\ket{\phi},
        \ldots,
        K^n\ket{\phi}
    \right\},
\end{equation}
and the exact inverse-response vector
\begin{equation}
    \ket{y}:=K^+\ket{\phi},
    \qquad
    \chi_F=\|\ket{y}\|^2.
\end{equation}

 Let $P_{n}$ denote the orthogonal projector onto
$\mathcal{V}_{n}$ such that
\begin{equation}
    \ket{y}
    =
    P_{n}\ket{y}
    +
    \bigl(\mathbb{I}-P_{n}\bigr)\ket{y},
\end{equation}
where $ P_{n}\ket{y}\in\mathcal{V}_{n},$ $\bigl(\mathbb{I}-P_{n}\bigr)\ket{y}
    \in\mathcal{V}_{n}^{\perp}.$ Geometrically, $P_{n}\ket{y}$ is the vector in $\mathcal{V}_{n}$ that
best approximates $\ket{y}$:
\begin{equation}
    P_{n}\ket{y}
    =
    \underset{\ket{x}\in\mathcal{V}_{n}}
    {\operatorname{argmin}}
    \,
    \|\ket{y}-\ket{x}\|^{2}.
\end{equation}

With this motivation, we define the order-$n$ variational approximation to the fidelity susceptibility as
\begin{equation}
    \chi_{F}^{(n)}
    :=
    \left\|P_{n}\ket{y}\right\|^{2}.
\end{equation}

\begin{theorem}[Monotone lower bounds]
\label{thm:krylov-lower-bound}
The Krylov approximations $\chi_F^{(n)}$ satisfy
\begin{equation}
    \chi_F-\chi_F^{(n)}
    =
    \|(\mathbb{I}-P_n)\ket{y}\|^2
    \geq0,
\label{eq:projector}
\end{equation}
and form the monotone hierarchy
\begin{equation}
    0\leq
    \chi_F^{(1)}
    \leq
    \chi_F^{(2)}
    \leq\cdots\leq
    \chi_F .
\label{eq:monotone-hierarchy}
\end{equation}
\end{theorem}

\begin{proof}
Since $P_n$ is an orthogonal projector,
\begin{equation}
    \ket{y}
    =
    P_n\ket{y}
    +
    (\mathbb{I}-P_n)\ket{y},
\end{equation}
with mutually orthogonal components. Hence,
\begin{equation}
    \|\ket{y}\|^2
    =
    \|P_n\ket{y}\|^2
    +
    \|(\mathbb{I}-P_n)\ket{y}\|^2,
\end{equation}
which immediately gives Eq.~\eqref{eq:projector} and therefore
$\chi_F^{(n)}\leq\chi_F$.

Moreover, $\mathcal{V}_n\subseteq\mathcal{V}_{n+1}$, so
$P_{n+1}-P_n$ is the orthogonal projector onto
$\mathcal{V}_{n+1}\cap\mathcal{V}_n^\perp$. Consequently,
\begin{align}
    \chi_F^{(n+1)}-\chi_F^{(n)}
    &=
    \|P_{n+1}\ket{y}\|^2-\|P_n\ket{y}\|^2
    \nonumber\\
    &=
    \|(P_{n+1}-P_n)\ket{y}\|^2
    \geq0,
\end{align}
which proves the monotone hierarchy.
\end{proof}

For a fixed order $n\geq1$, define the Krylov vectors
\begin{equation}
    \ket{v_i}:=K^{i+1}\ket{\phi},
    \qquad i=0,\ldots,n-1,
\end{equation}
which span $\mathcal V_n$. Their Gram matrix and the associated moment
vector \cite{strakosGeneGolubGerard2011} are
\begin{equation}
    (\Gamma_n)_{ij}
    :=
    \braket{v_i|v_j}
    =
    M_{i+j+2},
    \qquad
    (m_n)_i:=M_i,
\end{equation}
so that explicitly
\begin{equation}
\Gamma_n=
\begin{pmatrix}
M_2 & M_3 & \cdots & M_{n+1}\\
M_3 & M_4 & \cdots & M_{n+2}\\
\vdots & \vdots & \ddots & \vdots\\
M_{n+1} & M_{n+2} & \cdots & M_{2n}
\end{pmatrix},
\qquad
m_n=
\begin{pmatrix}
M_0\\ M_1\\ \vdots\\ M_{n-1}
\end{pmatrix}.
\end{equation}

\begin{theorem}[Moment representation]
\label{thm:moment-krylov}
If $\Gamma_n$ is nonsingular, the order-$n$ Krylov approximation is
\begin{equation}
    \chi_F^{(n)}
    =
    m_n^\dagger\Gamma_n^{-1}m_n .
\label{eq:krylov-moment-reconstruction}
\end{equation}
Hence, $\chi_F^{(n)}$ is completely determined by the response moments
$M_0,\ldots,M_{2n}$.
\end{theorem}

\begin{proof}
Expand the projected inverse-response vector as
\begin{equation}
    P_n\ket{y}
    =
    \sum_{j=0}^{n-1}c_j\ket{v_j}.
\end{equation}
The projection condition $\braket{v_i|y-P_ny}=0$ gives
\begin{equation}
    \sum_{j=0}^{n-1}
    \braket{v_i|v_j}c_j
    =
    \braket{v_i|y}.
\end{equation}
Since $\braket{v_i|y}
    =
    \bra{\phi}K^{i+1}K^+\ket{\phi}
    =
    M_i,$ where $KK^+\ket{\phi}=\ket{\phi}$, the coefficients satisfy
\begin{equation}
    \Gamma_n c=m_n.
\end{equation}
For nonsingular $\Gamma_n$, $c=\Gamma_n^{-1}m_n$. Therefore,
\begin{align}
    \chi_F^{(n)}
    &=
    \|P_n\ket{y}\|^2
    =
    c^\dagger\Gamma_n c
    \nonumber\\
    &=
    m_n^\dagger\Gamma_n^{-1}m_n,
\end{align}
which proves the result.
\end{proof}

For example, up to order-$2$ approximation, we have:
\begin{align}
    \chi_F^{(1)}
    &=
    \frac{M_0^2}{M_2}, \nonumber \\
      \chi_F^{(2)}
    &=
    \frac{
        M_0^2M_4-2M_0M_1M_3+M_1^2M_2
    }{
        M_2M_4-M_3^2
    }.
\end{align}

\subsection{Convergence analysis}

The convergence is controlled not by the full spectrum, but by the part of the spectrum probed by $\ket{\phi}$:
\begin{equation}
  \Delta_\phi
  :=
  \min_{\substack{j\geq1\\ \langle\Psi_j|\phi\rangle\neq0}}
  (E_j-E_0),
  \qquad
  \Lambda_\phi
  :=
  \max_{\substack{j\geq1\\ \langle\Psi_j|\phi\rangle\neq0}}
  (E_j-E_0).
\end{equation}
Since
$\Delta_H\leq\Delta_\phi\leq\Lambda_\phi\leq W_H$, the response condition number
can be substantially smaller than the Hamiltonian's condition number
\begin{equation}
  \kappa_\phi
  :=
  \frac{\Lambda_\phi}{\Delta_\phi}
  \leq
  \kappa_H
  :=
  \frac{W_H}{\Delta_H}.
\end{equation}
\begin{theorem}[Geometric convergence]
\label{thm:geometric-convergence-main}
For $\kappa_\phi > 1$, the Krylov hierarchy satisfies
\begin{equation}
  \frac{\chi_F-\chi_F^{(n)}}{\chi_F}
  \leq
  4\bigl(1+n\sqrt{\kappa_\phi}\bigr)^2
  \left(
    \frac{\sqrt{\kappa_\phi}-1}
         {\sqrt{\kappa_\phi}+1}
  \right)^{2n}.
\label{eq:geometric-convergence_1}
\end{equation}
\end{theorem}
The special case $\kappa_\phi=1$ corresponds to
$\chi_F^{(1)}=\chi_F$.
\begin{corollary}[Krylov-depth budget]
\label{cor:krylov-depth-main}
To guarantee
\begin{equation}
  \frac{\chi_F-\chi_F^{(n)}}{\chi_F}\leq\epsilon,
\end{equation}
it is sufficient to choose
\begin{equation}
  n
  =
  \mathcal{O}\!\left[
    \sqrt{\kappa_\phi}
    \left(
      \ln\kappa_\phi+\ln\epsilon^{-1}
    \right)
  \right]
  =
  \widetilde{\mathcal O}(\sqrt{\kappa_\phi}).
\end{equation}
\end{corollary}
The proof of Theorem \ref{thm:geometric-convergence-main} and Corollary \ref{cor:krylov-depth-main} is given in Appendix~\ref{sec:convergence}. These bounds also provide the budget for moment order, since the order-$n$ reconstruction requires moments up to $M_{2n}$. As this is a worst-case bound, for systems with spectrally sparse response, it can converge at low Krylov order.

\section{Classical-Shadow Certification}
\label{sec:classical-shadows}

We saw in the previous section, that the order-$n$ Krylov reconstruction requires the moments
$M_0,\ldots,M_{2n}$. We assume that $E_0$ and $\langle H_I\rangle$ are known aprior.
The required moments can be written as expectation values
\begin{equation}
    M_k
    =
    \Tr(\rho_0 O_k)
    =
    \bra{\phi}K^k\ket{\phi},
    \qquad
    O_k:=BK^kB .
    \label{eq:sh-moments-as-observables}
\end{equation}
Since each $M_k$ is a linear functional of the reference ground-state $\rho_0$, it can be estimated via classical shadow data set.

For numerical stability, we normalize using a scale
$\Omega\geq\|K\|_\infty$ and define
\begin{equation}
    \widetilde K:=\frac{K}{\Omega},
    \qquad
    \widetilde O_k:=B\widetilde K^kB,
    \qquad
    \widetilde M_k:=\Tr(\rho_0\widetilde O_k).
\end{equation}
Then $\|\widetilde O_k\|_\infty
    \leq
    \|B\|_\infty^2 ,$ independently of $k$. Since $\widetilde K^+=\Omega K^+$,
\begin{equation}
    \chi_F
    =
    \Omega^{-2}\widetilde\chi_F .
    \label{eq:cs-restore-units}
\end{equation}

\subsection{Simultaneous Shadow Estimation}

For our purpose, we can simply choose independent random single-qubit Pauli measurements. The details are reviewed in
Appendix~\ref{sec:classical_shadows}. For each snapshot
$\widehat\rho^{(t)}$, the estimator is unbiased
\begin{equation}
    X_k^{(t)}
    :=
    \Tr\!\left(
        \widetilde O_k\widehat\rho^{(t)}
    \right), \qquad \mathbb E[X_k^{(t)}]
    =
    \widetilde M_k .
\end{equation}
Importantly, the same set of snapshots can be reused for all
$2n+1$ moments. Suppose we have bounds 
\begin{equation}
  \Var(X_k^{(t)})\leq\nu_k, \qquad   \nu_{\max}^{(n)}
    :=
    \max_{0\leq k\leq2n}\nu_k .
\end{equation}
then, the classical shadows protocol \cite{Huang_2020} implies there exists a universal constant $C$ such that the number of shots
\begin{equation}
    N_{\rm sh}
    \geq
    C\,
    \frac{\nu_{\max}^{(n)}}{\eta^2}
    \log\!\frac{2(2n+1)}{\delta}
    \label{eq:sh-shot-count}
\end{equation}
suffices to guarantee
\begin{equation}
    \left|
        \widehat{\widetilde M}_k-\widetilde M_k
    \right|
    \leq\eta,
    \qquad
    0\leq k\leq2n,
    \label{eq:sh-uniform-moment-error}
\end{equation}
simultaneously with probability at least $1-\delta$. Thus, the dependence on the number of required moments is only
logarithmic.

\subsection{Certified Reconstruction}
Since the true fidelity susceptibility is unknown, we cannot work with the relative accuracy used in Theorem \ref{thm:geometric-convergence-main}. Rather, we work with additive moment accuracy. From the estimated normalized moments, we define
\begin{equation}
    (\widehat{\widetilde m}_n)_i
    :=
    \widehat{\widetilde M}_i,
    \qquad
    (\widehat{\widetilde\Gamma}_n)_{ij}
    :=
    \widehat{\widetilde M}_{i+j+2},
    \label{eq:sh-estimated-hankel}
\end{equation}
and
\begin{equation}
    \widehat{\widetilde\chi}_F^{(n)}
    :=
    \widehat{\widetilde m}_n^\dagger
    \widehat{\widetilde\Gamma}_n^{-1}
    \widehat{\widetilde m}_n .
\end{equation}

\begin{theorem}[Certified fidelity-susceptibility bound]
\label{thm:certified-shadow-bound}
Assume Eq.~\eqref{eq:sh-uniform-moment-error} holds and define $\widehat\gamma_n
    :=
    \lambda_{\min}
    (\widehat{\widetilde\Gamma}_n)$
For every order satisfying $\widehat\gamma_n>n\eta,$ the fidelity susceptibility obeys, with probability at least $1-\delta$,
\begin{equation}
    \chi_F
    \geq
    \mathcal B(n)
    :=
    \Omega^{-2}
    \left[
        \widehat{\widetilde\chi}_F^{(n)}
        -
        L_n\eta
        -
        R_n\eta^2
    \right].
\label{eq:cs-main-certified-bound}
\end{equation}
\end{theorem}
The proof is in Appendix~\ref{sec:classical_shadows}. The parameters entering the bound are all shadow data-dependent quantities
\begin{align}
    \widehat c_n
    &:=
    \widehat{\widetilde\Gamma}_n^{-1}
    \widehat{\widetilde m}_n,
    \\
    \overline c_n
    &:=
    \frac{
        \|\widehat c_n\|_2
        +
        \sqrt n\,\eta/\widehat\gamma_n
    }{
        1-n\eta/\widehat\gamma_n
    },
    \\
    L_n
    &:=
    2\sqrt n\,\overline c_n
    +
    n\overline c_n^2,
    \\
    R_n
    &:=
    \frac{
        (\sqrt n+n\overline c_n)^2
    }{
        \widehat\gamma_n
    }.
\end{align}
The terms $L_n\eta$ and $R_n\eta^2$ are finite-sample reconstruction
penalties arising from errors in the estimated moment vector and Hankel
matrix. Among all admissible orders, the tightest certified bound is
\begin{equation}
    \mathcal B^\star
    :=
    \max_{n\in\mathcal A}\mathcal B(n),
    \qquad
    \mathcal A
    :=
    \{n:\widehat\gamma_n>n\eta\}.
\end{equation}
The whole algorithm is outlined in Algorithm~\ref{alg:shadow-fidelity}. 

\subsection{Imperfect Ground-State Preparation}

So far, we have assumed the preparation of a perfect ground state
$\rho_0=\ket{\Psi_0}\!\bra{\Psi_0}$. Suppose instead that
\begin{equation}
    \frac12
    \|\rho_{\rm p}-\rho_0\|_1
    \leq
    \epsilon_{\rm p}.
\end{equation}
Then
\begin{equation}
    \left|
        \Tr[
            (\rho_{\rm p}-\rho_0)\widetilde O_k
        ]
    \right|
    \leq
    2\epsilon_{\rm p}\|B\|_\infty^2 .
\end{equation}
Therefore, the same certification applies after replacing
\begin{equation}
    \eta
    \longrightarrow
    \eta_{\rm tot}
    :=
    \eta
    +
    2\epsilon_{\rm p}\|B\|_\infty^2.
\end{equation}

\begin{algorithm}[t]
\caption{Certified shadow estimation of $\chi_F$}
\label{alg:shadow-fidelity}
\DontPrintSemicolon
\KwIn{$H$, $H_I$, $E_0$, $\langle H_I\rangle$, scale
$\Omega\geq\|H-E_0\|_\infty$, maximum order $n_{\max}$,
moment accuracy $\eta$, confidence $1-\delta$.}
\KwOut{Certified lower bound $\mathcal B^\star$ on $\chi_F$.}

Set
$B\leftarrow H_I-\langle H_I\rangle I$ and
$\widetilde K\leftarrow(H-E_0I)/\Omega$\;

Construct $\widetilde O_k=B\widetilde K^kB$ for
$k=0,\ldots,2n_{\max}$\;

Choose $N_{\rm sh}$ according to
Eq.~\eqref{eq:sh-shot-count}\;

Collect $N_{\rm sh}$ local-Pauli classical-shadow snapshots\;

Estimate
$\widehat{\widetilde M}_0,\ldots,
\widehat{\widetilde M}_{2n_{\max}}$
simultaneously\;

$\mathcal A\leftarrow\varnothing$\;

\For{$n=1,\ldots,n_{\max}$}{
    Assemble
    $\widehat{\widetilde m}_n$ and
    $\widehat{\widetilde\Gamma}_n$\;

    Compute
    $\widehat\gamma_n
    =
    \lambda_{\min}(\widehat{\widetilde\Gamma}_n)$\;

    \If{$\widehat\gamma_n>n\eta$}{
        Solve
        $\widehat{\widetilde\Gamma}_n\widehat c_n
        =
        \widehat{\widetilde m}_n$\;

        Compute
        $\widehat{\widetilde\chi}_F^{(n)}
        =
        \widehat{\widetilde m}_n^\dagger\widehat c_n$\;

        Evaluate $\mathcal B(n)$ from
        Eq.~\eqref{eq:cs-main-certified-bound}\;

        $\mathcal A\leftarrow\mathcal A\cup\{n\}$\;
    }
}

\Return
$\mathcal B^\star
=
\max_{n\in\mathcal A}\mathcal B(n)$\;
\end{algorithm}

\section{Applications}
\label{subsec:entanglement_verification}
We will now present two applications of our certified bound protocol. First, in quantum metrology, and second, by extending the construction to linear static susceptibilities.

\subsection{Quantum Fisher Information}
For a pure ground state generated by the parameter variation
$H_I=\partial_\lambda H$, the quantum Fisher information (QFI) satisfies $ F_\lambda = 4\chi_F$ \cite{Liu_2019,Giovannetti_2006}.  Therefore, the certified fidelity-susceptibility bound immediately gives
\begin{equation}
    F_\lambda \geq 4\mathcal{B}^\star ,
\end{equation}
providing a statistically certified lower bound on the QFI.

\subsection{Extension to Linear Static Susceptibility}
\label{subsec:static_susceptibility}

Our construction can readily be extended to linear static
susceptibilities, which quantify the response of a system to a static external perturbation
\cite{Kubo:1957mj, Stanley1972IntroductionTP, RevModPhys.70.653}.
Consider $ H(h)=H-hO ,$ where $O$ is an observable coupled to an external field $h$. For a
nondegenerate ground state $\ket{\Psi_0}$, the zero-temperature static susceptibility $ \chi_O
    :=
    \left.
    \frac{\partial \langle O\rangle_h}{\partial h}
    \right|_{h=0}$
has the spectral representation
\begin{equation}
    \chi_O
    =
    2\sum_{j>0}
    \frac{
        \left|
            \bra{\Psi_j}O\ket{\Psi_0}
        \right|^2
    }{
        E_j-E_0
    }.
    \label{eq:static_susceptibility_spectral}
\end{equation}
Defining $B_O:=O-\langle O\rangle I,$ and $\ket{\phi_O}:=B_O\ket{\Psi_0}$, and $K=H-E_0I$, Eq.~\eqref{eq:static_susceptibility_spectral} becomes
\begin{equation}
    \chi_O
    =
    2\bra{\phi_O}K^+\ket{\phi_O}.
\end{equation}
Hence, the susceptibility is determined by the negative first moment, rather than the negative second moment appearing in the fidelity susceptibility.

The corresponding order-$n$ Krylov construction uses
\begin{equation}
    \mathcal{V}_n^{(O)}
    =
    \operatorname{span}
    \left\{
        \ket{\phi_O},
        K\ket{\phi_O},
        \ldots,
        K^{n-1}\ket{\phi_O}
    \right\},
\end{equation}
with moments $   M_k^{(O)}
    :=
    \bra{\phi_O}K^k\ket{\phi_O}$. The associated Hankel matrix is therefore shifted by one power,
\begin{equation}
    \left(\Gamma_n^{(O)}\right)_{ij}
    =
    M_{i+j+1}^{(O)},
    \qquad i,j=0,\ldots,n-1,
\end{equation}
and the similar variational construction gives a monotone lower bound
hierarchy for $\chi_O$, with an overall factor of two. In addition, the standard Krylov convergence applies without the polynomial prefactor appearing in
Eq.~\eqref{eq:geometric-convergence_1}
\cite{10.1093/acprof:oso/9780199655410.001.0001}.

\section{Conclusion}

By integrating Krylov subspace methods with resolvent-based reformulation of fidelity susceptibility, we have developed a framework for estimating and certifying the fidelity susceptibility from randomized measurements of a ground state. This produces a hierarchy
\begin{equation}
0\leq \chi_F^{(1)} \leq \chi_F^{(2)} \leq \cdots \leq \chi_F ,
\end{equation}
whose order-$n$ is determined by the moments $M_0,\ldots,M_{2n}$. These moments are linear expectation values of the same ground-state and can therefore be estimated simultaneously using randomized single-copy measurements. This eliminates the need for an excited-state spectrum, which would not be practical for the classical shadows protocol. We illustrate the mechanics for the transverse-field Ising model. Beyond fidelity susceptibility itself, we also provided how the certified lower bound has direct operational applications in metrology and related quantities.

Several interesting future directions remain. In this work, we have restricted ourselves to $f(x) = 1/x^2$ for the fidelity susceptibility and $1/x$ for the linear static susceptibility. Krylov subspace methods can be applied to a general $f(x)$, thereby allowing a large class of observables to be estimated from classical shadows. Combining this with Krylov shadow tomography \cite{zhangKrylovShadowTomography2025a} would make the framework far more general. More technical work would focus on developing adaptive criteria for choosing the Krylov order from experimental data and extending the framework to finite-temperature and mixed-state susceptibilities. 

The next crucial question is how to control sample complexity, as it depends on the ensembles selected for randomized measurements. If the ensemble is random single-qubit Pauli measurements (as done in this work), the sample complexity depends on the locality of the observable $X_k^{(t)}$, which typically increases as $k$ increases. If, instead, the ensemble consists of Random $n$-qubit Clifford circuits, then the dependence on locality drops. So, we have a trade-off, and proper analysis is left for future work. Finally, if the goal is just to compute the Fidelity susceptibility rather than arbitrary observables, de-randomization techniques are recommended \cite{Huang_2021}, as they can substantially reduce sample complexity.

\section*{Acknowledgements}
 The research is part of the Munich Quantum Valley, which is supported by the Bavarian state government with funds from the Hightech Agenda Bayern Plus. During the preparation of this article, the author used ChatGPT (OpenAI) and Claude (Anthropic) to assist with mathematical cross-checking, literature discovery, and language editing. All mathematical statements, derivations, references, and conclusions were independently verified by the author, who takes full responsibility for the content.

\bibliographystyle{apsrev4-2}
\bibliography{bib.bib}

\appendix
\onecolumngrid


\section{Fidelity susceptibility in Classical Shadow Friendly form}

\subsection{Resolvent form}

\label{sec:spectral_representation}

We now derive the two equivalent expressions for fidelity susceptibility $\chi_F$. Differentiating the instantaneous eigenvalue equation $H(\lambda)|\Psi_0(\lambda)\rangle
=
E_0(\lambda)|\Psi_0(\lambda)\rangle$ with respect to $\lambda$ gives
\begin{align}
\partial_\lambda H|\Psi_0\rangle
+
H|\partial_\lambda\Psi_0\rangle
&=
\partial_\lambda E_0|\Psi_0\rangle
+
E_0|\partial_\lambda\Psi_0\rangle \nonumber \\
H_I|\Psi_0\rangle
+
H|\partial_\lambda\Psi_0\rangle
&=
(\partial_\lambda E_0)|\Psi_0\rangle
+
E_0|\partial_\lambda\Psi_0\rangle .
\label{eq:differentiated_ground_state_equation_HI}
\end{align}
where, we used $H_I=\partial_\lambda H$. Projecting Eq.~\eqref{eq:differentiated_ground_state_equation_HI} onto an excited eigenstate $\langle \Psi_j|$ with $j\neq 0$, we obtain
\begin{align}
\langle\Psi_j|H_I|\Psi_0\rangle
+
\langle\Psi_j|H|\partial_\lambda\Psi_0\rangle
&=
(\partial_\lambda E_0)\langle\Psi_j|\Psi_0\rangle
+
E_0\langle\Psi_j|\partial_\lambda\Psi_0\rangle \nonumber \\
\langle\Psi_j|H_I|\Psi_0\rangle
+
E_j\langle\Psi_j|\partial_\lambda\Psi_0\rangle
&=
E_0\langle\Psi_j|\partial_\lambda\Psi_0\rangle .
\label{eq:projection_excited_state_2}
\end{align}
where, we used $\langle\Psi_j|\Psi_0\rangle=0$ for $j\neq 0$ and $\langle\Psi_j|H=E_j\langle\Psi_j|$ in the second line. Therefore,
\begin{equation}
\langle\Psi_j|\partial_\lambda\Psi_0\rangle
=
-\frac{\langle\Psi_j|H_I|\Psi_0\rangle}{E_j-E_0},
\qquad
j\neq 0.
\label{eq:derivative_component}
\end{equation}

We define the ground-state projector and its orthogonal complement by $P
= |\Psi_0\rangle\langle\Psi_0|,$ and $P^\perp = I-P = \sum_{j=1}^{D-1}
|\Psi_j\rangle\langle\Psi_j|$ respectively. The fidelity susceptibility is the squared norm of the component of $|\partial_\lambda\Psi_0\rangle$ orthogonal to the ground state:
\begin{align}
\chi_F
&= \langle\partial_\lambda\Psi_0|P^\perp|\partial_\lambda\Psi_0\rangle \nonumber \\
&= \langle\partial_\lambda\Psi_0|
\left(
\sum_{j=1}^{D-1}
|\Psi_j\rangle\langle\Psi_j|
\right)
|\partial_\lambda\Psi_0\rangle
\nonumber\\
&=
\sum_{j=1}^{D-1}
\langle\partial_\lambda\Psi_0|\Psi_j\rangle
\langle\Psi_j|\partial_\lambda\Psi_0\rangle
\nonumber\\
&=
\sum_{j=1}^{D-1}
\left|
\langle\Psi_j|\partial_\lambda\Psi_0\rangle
\right|^2 \nonumber \\
&= \sum_{j=1}^{D-1}
\frac{
|\langle\Psi_j|H_I|\Psi_0\rangle|^2
}{
(E_j-E_0)^2
}.
\label{eq:chi_derivative_expansion}
\end{align}

\subsection{Reduction of Fidelity Susceptibility to a Classical Moment-Reconstruction Problem}
\label{subsec:fs_moment_reduction}

Let $ B
    :=
    H_I
    -
    \bra{\Psi_0}H_I\ket{\Psi_0}I,$ and $\ket{\phi}
    :=
    B\ket{\Psi_0}$. Then, Eq.~\eqref{eq:chi_derivative_expansion} can be written as
\begin{align}
    \chi_F
    &=
    \bra{\Psi_0}
    B\left(K^{+}\right)^2B
    \ket{\Psi_0}
    \nonumber\\
    &=
    \bra{\phi}
    \left(K^{+}\right)^2
    \ket{\phi}.
    \label{eq:fs_pseudoinverse_form}
\end{align}
where, $K:=H-E_0 I$. Thus, fidelity susceptibility is the expectation value of the non-polynomial spectral function $f(x)=x^{-2}$ on the excited-state support of $K$. We now define the positive spectral measure
\begin{equation}
    d\mu(E)
    :=
    \sum_{j\neq0}
    \left|
        \bra{\Psi_j}H_I\ket{\Psi_0}
    \right|^2
    \delta\!\left(
        E-\left(E_j-E_0\right)
    \right)dE.
    \label{eq:response_spectral_measure}
\end{equation}
with support $\operatorname{supp}(\mu)
    \subseteq
    [\Delta,\Lambda]$, where $\Delta
    :=
    E_1-E_0
    >
    0$ is the spectral gap and $\Lambda
    :=
    E_{\max}-E_0$ is the maximum excitation energy. For every $k\in\mathbb{Z}$, we define the response moment $M_k$ as
\begin{align}
    M_k
    &=\bra{\phi}K^k\ket{\phi} \nonumber \\
    &=  \bra{\Psi_0}BK^kB\ket{\Psi_0} \nonumber \\
   &=  \sum_{i,j\neq0}
    \bra{\Psi_0}H_I\ket{\Psi_i}
    \bra{\Psi_i}K^k\ket{\Psi_j}
    \bra{\Psi_j}H_I\ket{\Psi_0}
    \nonumber\\
    &=
    \sum_{j\neq0}
    \left|
        \bra{\Psi_j}H_I\ket{\Psi_0}
    \right|^2
    \left(E_j-E_0\right)^k.
    \label{eq:response_moment_spectral}
\end{align}

Multiplying spectral measure Eq.~\eqref{eq:response_spectral_measure} by $E^k$ and integrating between $\Delta$ and $\Lambda$ gives the moment $M_k$ as:
\begin{align}
   \int_\Delta^\Lambda E^k \, d\mu(E) &=   \int_\Delta^\Lambda E^k \sum_{j\neq 0} |\langle \Psi_j | H_I | \Psi_0 \rangle|^2 \,
   \delta\!\big(E - (E_j - E_0)\big)\, dE \nonumber \\
   &= \sum_{j\neq 0} |\langle \Psi_j | H_I | \Psi_0 \rangle|^2 \, (E_j - E_0)^k \nonumber \\
   &= M_k
\end{align}
where we used the shift property of delta function $\int g(E) \, \delta(E - a) \, dE = g(a)$ in the second line.

Similarly, the fidelity susceptibility  Eq.~\eqref{eq:chi_derivative_expansion} is just the negative second moment $M_{-2}$:
\begin{align}
    \chi_F
    &=\sum_{j=1}^{D-1}
\frac{
|\langle\Psi_j|H_I|\Psi_0\rangle|^2
}{
(E_j-E_0)^2
} \nonumber \\
    &= \sum_{j\neq0}
    \left|
        \bra{\Psi_j}H_I\ket{\Psi_0}
    \right|^2
    \left(E_j-E_0\right)^{-2}
    \nonumber\\
    &=
    \int_{\Delta}^{\Lambda}
    E^{-2}\,d\mu(E) \nonumber \\
    &= M_{-2}
    \label{eq:fs_negative_moment}
\end{align}
 Therefore, the central problem is the following moment-reconstruction task:

\begin{equation}
    \left\{
        M_k
        =
        \int_{\Delta}^{\Lambda}
        E^k\,d\mu(E)
    \right\}_{k=0}^{k_{\max}}
    \quad
    \longrightarrow
    \quad
    \chi_F
    =
    \int_{\Delta}^{\Lambda}
    E^{-2}\,d\mu(E).
\end{equation}
In particular, if $p_n(E)$ is a degree-$n$ polynomial approximating $E^{-2}$ on $[\Delta,\Lambda]$, $  p_n(E)
    =
    \sum_{k=0}^{n}c_kE^k,$ then
\begin{align}
    \chi_F
    &=
    \int_{\Delta}^{\Lambda}
    E^{-2}\,d\mu(E)
    \nonumber\\
    &\approx
    \int_{\Delta}^{\Lambda}
    p_n(E)\,d\mu(E)
    \nonumber\\
    &=
    \sum_{k=0}^{n}
    c_k
    \int_{\Delta}^{\Lambda}
    E^k\,d\mu(E)
    \nonumber\\
    &=
    \sum_{k=0}^{n}c_kM_k.
    \label{eq:fs_polynomial_moment_estimator}
\end{align}

\begin{figure}
    \centering
    \includegraphics[width=\linewidth]{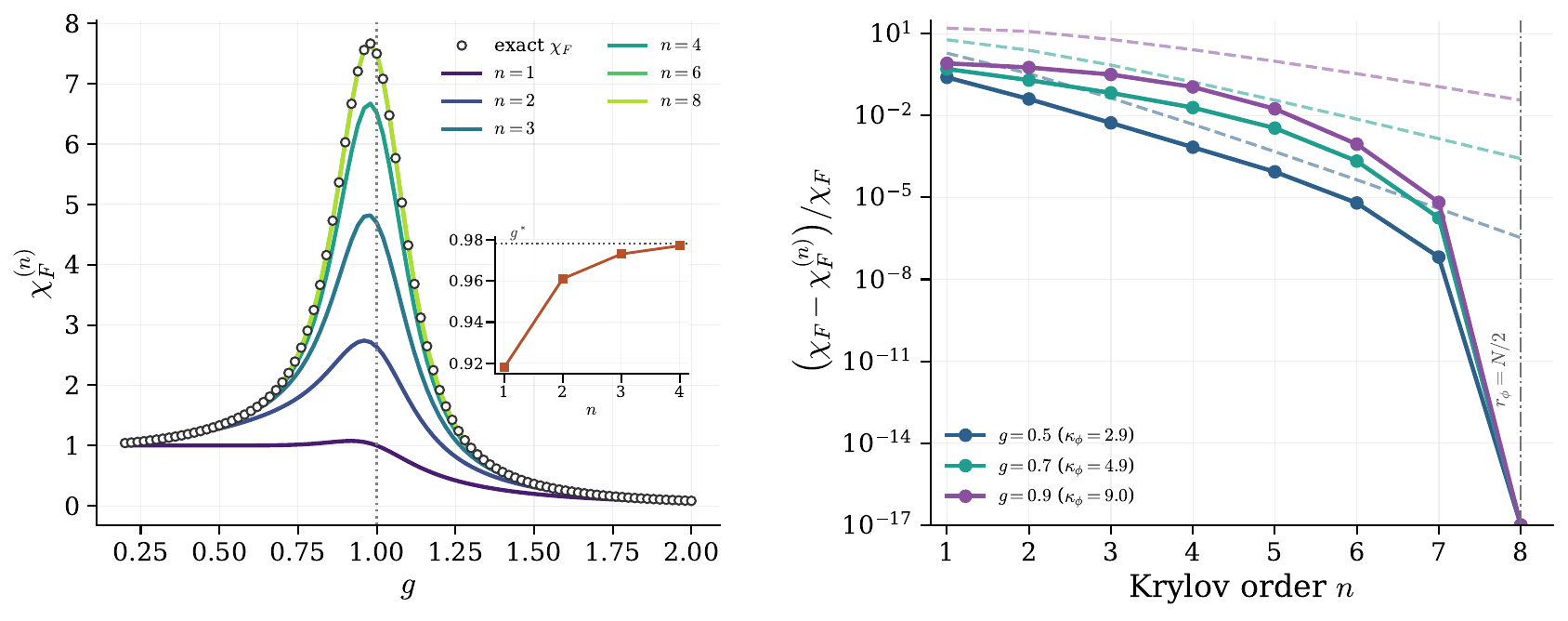}
\caption{Monotonic construction for the $N = 16$ 1D TFIM. \textbf{Left:} $\chi_F^{(n)} $ compared with $\chi_F$. The inset shows $g_n^*$ approaching $g^*=0.978$. \textbf{Right:} relative truncation error versus Krylov order. The errors remain below the geometric bound of Theorem~\ref{thm:geometric-convergence-main} and vanish at $n=N/2=8$, where the Krylov hierarchy terminates.}
    \label{fig:tfim_algebra}
\end{figure}

In figure \ref{fig:tfim_algebra}, we illustrate the moment--Krylov construction for the one-dimensional transverse-field Ising model (1D TFIM) with periodic boundary conditions,
\begin{equation}
    H(g)=-\sum_i Z_iZ_{i+1}-g\sum_i X_i,
    \qquad
    H_I=-\sum_i X_i .
\end{equation}
The response contains only $N/2$ distinct active gaps; the hierarchy terminates exactly at $n = N/2$. Away from the critical region, only a few Krylov orders are needed, while convergence slows near $g=1$. Interestingly, already at $n=2$ the finite-size transition is located within about $2\%$. Thus, low-order approximations may be sufficient when the goal is to detect the transition rather than reconstruct the full susceptibility.


\section{Convergence rate analysis}
\label{sec:convergence}

 Alongside the many-body gap and bandwidth,
\begin{equation}
  \Delta_H:=\min_{j\ge1}\bigl(E_j-E_0\bigr),
  \qquad
  W_H:=\max_{j}\bigl(E_j-E_0\bigr),
  \label{eq:gap-bandwidth}
\end{equation}
we define
\begin{equation}
  \Delta_\phi:=\!\!\min_{\substack{j\ge1\\ \braket{\Psi_j|B|\Psi_0}\neq0}}\!\!
  \bigl(E_j-E_0\bigr),
  \qquad
  \Lambda_\phi:=\!\!\max_{\substack{j\ge1\\ \braket{\Psi_j|B|\Psi_0}\neq0}}\!\!
  \bigl(E_j-E_0\bigr),
  \label{eq:response-edges}
\end{equation}
together with the response condition number $\kappa_\phi:=\frac{\Lambda_\phi}{\Delta_\phi}\;\ge\;1$. 
They satisfy: $\Delta_H\le\Delta_\phi\le\Lambda_\phi\le W_H,$ hence, $\kappa_\phi\;\le\;\kappa_H:=\frac{W_H}{\Delta_H}.$ By construction, 
\begin{equation}
    \operatorname{supp}\mu\subseteq[\Delta_\phi,\Lambda_\phi].
\label{eq:support}
\end{equation}

The reason for defining these quantities is that the convergence depends on the part of the spectrum that the perturbation actually probes. The degenerate case $\kappa_\phi=1$ corresponds to a response confined to a
single excitation energy, for which $\chi_F^{(1)}=\chi_F$ trivially.

\begin{lemma}[Residual representation]
\label{lem:residual}
Let
\begin{equation}
  \mathcal{Q}_{n+1}:=\bigl\{q\in\mathbb{R}[x]:\ \deg q\le n+1,\ q(0)=1,\
  q'(0)=0\bigr\}
  \label{eq:admissible-class}
\end{equation}
be the class of admissible residual polynomials. Then the order-$n$
truncation error admits the variational representation
\begin{equation}
  \chi_F-\chi_F^{(n)}
  =\min_{q\in\mathcal{Q}_{n+1}}
   \int_{\Delta_\phi}^{\Lambda_\phi}x^{-2}\,q(x)^2\,d\mu(x).
  \label{eq:residual}
\end{equation}
\end{lemma}

\begin{proof}
Let $\ket{y}:=K^{+}\ket{\phi}$; since
$\ket{\phi}$ lies in the orthogonal complement of the ground state,
$K\ket{y}=\ket{\phi}$. Furthermore, from Eq. \ref{eq:projector}, 
\begin{equation}
  \chi_F-\chi_F^{(n)}
  =\bigl\|(\mathbb{I}-P_n)\ket{y}\bigr\|^{2}
  =\min_{\ket{v}\in\mathcal{V}_n}\bigl\|\ket{y}-\ket{v}\bigr\|^{2}.
  \label{eq:projection-residual}
\end{equation}
Every $\ket{v}\in\mathcal{V}_n$ can be written as
\begin{equation}
   \ket{v}=\sum_{j=1}^{n}c_jK^{j}\ket{\phi}=\sum_{j=1}^{n}c_jK^{j+1}\ket{y} 
\end{equation}
such that, 
\begin{equation}
    \ket{y}-\ket{v}=q(K)\ket{y}, \qquad q(x)=1-\sum_{j=1}^{n}c_j\,x^{\,j+1}.
\label{eq:residual-poly-form}
\end{equation}
$q(x)$ are thus the polynomials of degree at most $n+1$ with $q(0)=1$ and $q'(0)=0$.  Conversely, every polynomial satisfying these conditions can be written
in the form \eqref{eq:residual-poly-form}; the correspondence is bijective.

It is sufficient to minimize over real polynomials. Indeed, for a
complex polynomial $q=q_{\mathrm R}+iq_{\mathrm I}$ satisfying the same
constraints,
\begin{equation}
    \int x^{-2}|q(x)|^2\,d\mu(x)
    =
    \int x^{-2}
    \left[
        q_{\mathrm R}(x)^2+q_{\mathrm I}(x)^2
    \right]d\mu(x)
    \geq
    \int x^{-2}q_{\mathrm R}(x)^2\,d\mu(x),
\end{equation}
while $q_{\mathrm R}\in\mathcal{Q}_{n+1}$. Finally, the spectral representation of
$K$ gives
\begin{equation}
  \bigl\|q(K)\ket{y}\bigr\|^{2}
  =\sum_{m\neq0}\frac{|\braket{\Psi_m|\phi}|^{2}}{(E_m-E_0)^{2}}\,
   q(E_m-E_0)^{2}
  =\int_{\Delta_\phi}^{\Lambda_\phi}x^{-2}q(x)^{2}\,d\mu(x),
\end{equation}
which combined with Eq.~\eqref{eq:projection-residual} proves
Eq.~\eqref{eq:residual}.
\end{proof}

\subsection{Geometric convergence}

The relative error decreases exponentially in $n$, up to a polynomial
factor. The argument is the standard Chebyshev minimax construction,
modified to accommodate the constraint $q'(0)=0$ that arises because
$\mathcal{V}_n$ is generated from $K\ket{\phi}$ rather than $\ket{\phi}$.

\begin{theorem}[Geometric convergence]
\label{thm:geometric-convergence}
Assume Eq.~\eqref{eq:support} with $\kappa_\phi>1$, and set
\begin{equation}
  \theta_\phi:=\ln\!\left(\frac{\sqrt{\kappa_\phi}+1}
                               {\sqrt{\kappa_\phi}-1}\right)
  =2\operatorname{artanh}\bigl(\kappa_\phi^{-1/2}\bigr),
  \qquad
  \rho_\phi:=e^{-\theta_\phi}
  =\frac{\sqrt{\kappa_\phi}-1}{\sqrt{\kappa_\phi}+1}.
  \label{eq:theta-rho}
\end{equation}
Then for every integer $n\ge1$,
\begin{equation}
  \frac{\chi_F-\chi_F^{(n)}}{\chi_F}
  \;\le\;\frac{\bigl(1+n\sqrt{\kappa_\phi}\bigr)^{2}}
              {\cosh^{2}(n\theta_\phi)}
  \;\le\;4\bigl(1+n\sqrt{\kappa_\phi}\bigr)^{2}\rho_\phi^{\,2n}.
  \label{eq:geometric-convergence}
\end{equation}
The hierarchy therefore converges geometrically in the Krylov order $n$, at
a rate fixed by the response condition number.
\end{theorem}

\begin{proof}
By Lemma~\ref{lem:residual}, for any $q\in\mathcal{Q}_{n+1}$,
\begin{align}
    \frac{\chi_F-\chi_F^{(n)}}{\chi_F}
    &\leq
    \frac{
        \displaystyle
        \int_{\Delta_\phi}^{\Lambda_\phi}
        x^{-2}q(x)^2\,d\mu(x)
    }{
        \displaystyle
        \int_{\Delta_\phi}^{\Lambda_\phi}
        x^{-2}\,d\mu(x)
    }
    \nonumber\\
    &\leq
    \left(
        \max_{x\in[\Delta_\phi,\Lambda_\phi]}
        |q(x)|
    \right)^2.
    \label{eq:uniform-residual-bound}
\end{align}
It is therefore sufficient to find a single polynomial that is uniformly small on the response interval. A natural candidate is then a polynomial from the Chebyshev families. The next step is to construct such polynomial.

The affine transformation
\begin{equation}
    w(x)
    :=
    \frac{
        \Lambda_\phi+\Delta_\phi-2x
    }{
        \Lambda_\phi-\Delta_\phi
    }.
    \label{eq:affine-map}
\end{equation}
satisfies
$w(\Delta_\phi)=1$ and $w(\Lambda_\phi)=-1$, hence sends
$[\Delta_\phi,\Lambda_\phi]$ onto $[-1,1]$, and
\begin{equation}
  z_\phi:=w(0)=\frac{\Lambda_\phi+\Delta_\phi}{\Lambda_\phi-\Delta_\phi}
  =\frac{\kappa_\phi+1}{\kappa_\phi-1}>1 .
  \label{eq:z-phi}
\end{equation}
Let $T_n$ denote the Chebyshev polynomial of the first kind and define
\begin{equation}
    \varphi_n(x)
    :=
    \frac{T_n\!\left(w(x)\right)}
         {T_n(z_\phi)}.
    \label{eq:normalized-chebyshev}
\end{equation}
Since $w(0)=z_\phi$, $ \varphi_n(0)=1.$ Furthermore, $|T_n(u)|\leq1$ for every $u\in[-1,1]$, and hence
\begin{equation}
  \max_{x\in[\Delta_\phi,\Lambda_\phi]}|\varphi_n(x)|\le\frac{1}{T_n(z_\phi)} .
  \label{eq:phi-uniform-bound}
\end{equation}
In general $\varphi_n'(0)\neq0$, so we introduce the corrected polynomial
\begin{equation}
  q_n(x):=\varphi_n(x)\bigl(1+\alpha_nx\bigr),
  \qquad \alpha_n:=-\varphi_n'(0).
  \label{eq:corrected-chebyshev}
\end{equation}
Then $q_n(0)=\varphi_n(0)=1$ and
$q_n'(0)=\varphi_n'(0)+\alpha_n\varphi_n(0)=0$, while
$\deg q_n\le n+1$; hence 
\begin{equation}
    q_n\in\mathcal{Q}_{n+1}.
\end{equation}

Now, our task is to bound $q_n$ using the usual Chebyshev identities. Since $z_\phi>1$ there is a unique $\theta_\phi>0$ with
$z_\phi=\cosh\theta_\phi$. The identities $T_n(\cosh\theta)=\cosh(n\theta)$
and $T_n'(\cosh\theta)=n\sinh(n\theta)/\sinh\theta$, together with
$\tanh(n\theta_\phi)\le1$ and
$\sinh\theta_\phi=\sqrt{z_\phi^{2}-1}=2\sqrt{\kappa_\phi}/(\kappa_\phi-1)$,
give
\begin{equation}
  \frac{T_n'(z_\phi)}{T_n(z_\phi)}
  =\frac{n\tanh(n\theta_\phi)}{\sinh\theta_\phi}
  \le\frac{n(\kappa_\phi-1)}{2\sqrt{\kappa_\phi}} .
  \label{eq:chebyshev-derivative-bound}
\end{equation}

Differentiating $\varphi_n$ with $w'=-2/(\Lambda_\phi-\Delta_\phi)$ yields
$\alpha_n=2\,T_n'(z_\phi)/[(\Lambda_\phi-\Delta_\phi)T_n(z_\phi)]$, and
inserting $\Lambda_\phi-\Delta_\phi=\Delta_\phi(\kappa_\phi-1)$ together
with Eq.~\eqref{eq:chebyshev-derivative-bound} gives
$0<\alpha_n\le n/(\Delta_\phi\sqrt{\kappa_\phi})$, i.e.
\begin{equation}
  \alpha_n\Lambda_\phi\le n\sqrt{\kappa_\phi}.
  \label{eq:alpha-lambda-bound}
\end{equation}
As $\alpha_n>0$, the factor $1+\alpha_nx$ increases on the response
interval and is maximal at $x=\Lambda_\phi$. Combining this with
Eq.~\eqref{eq:phi-uniform-bound} and $T_n(z_\phi)=\cosh(n\theta_\phi)$,
\begin{equation}
  \max_{x\in[\Delta_\phi,\Lambda_\phi]}|q_n(x)|
  \le\frac{1+\alpha_n\Lambda_\phi}{T_n(z_\phi)}
  \le\frac{1+n\sqrt{\kappa_\phi}}{\cosh(n\theta_\phi)} ,
\end{equation}
which substituted into Eq.~\eqref{eq:uniform-residual-bound} proves the
first inequality in Eq.~\eqref{eq:geometric-convergence}.

Using $e^{\theta_\phi}=z_\phi+\sqrt{z_\phi^{2}-1}$ and
Eq.~\eqref{eq:z-phi},
\begin{equation}
  e^{\theta_\phi}
  =\frac{\kappa_\phi+2\sqrt{\kappa_\phi}+1}{\kappa_\phi-1}
  =\frac{\bigl(\sqrt{\kappa_\phi}+1\bigr)^{2}}
        {\bigl(\sqrt{\kappa_\phi}+1\bigr)\bigl(\sqrt{\kappa_\phi}-1\bigr)}
  =\frac{\sqrt{\kappa_\phi}+1}{\sqrt{\kappa_\phi}-1},
  \label{eq:sqrt-mechanism}
\end{equation}
so $e^{-\theta_\phi}=\rho_\phi$ as claimed in Eq.~\eqref{eq:theta-rho}.
The elementary estimate $\cosh(n\theta_\phi)\ge\tfrac12e^{n\theta_\phi}$
then gives $\cosh^{-2}(n\theta_\phi)\le4\rho_\phi^{\,2n}$, which yields the
second inequality.
\end{proof}

\subsection{Worst-case budget}

We now provide a sufficient Krylov depth in the worst case; the depth actually
required may be considerably smaller.

\begin{corollary}[Krylov-depth and moment-order budget]
\label{cor:moment-budget}
Let $\epsilon\in(0,1)$ and $\kappa_\phi>1$. Every Krylov order
\begin{equation}
  n\;\ge\;n_\epsilon
  :=\left\lceil
      \frac{\sqrt{\kappa_\phi}}{2}\,
      \ln\!\left(\frac{16\kappa_\phi^{2}}{e^{2}\epsilon}\right)
    \right\rceil
  \label{eq:krylov-depth-budget}
\end{equation}
guarantees $(\chi_F-\chi_F^{(n)})/\chi_F\le\epsilon$, and the order-$n$
Hankel construction requires moments only up to order
\begin{equation}
  k_{\max}=2n .
  \label{eq:max-moment-order}
\end{equation}
The sufficient Krylov depth, therefore, scales as
$n_\epsilon=\mathcal{O}\bigl[\sqrt{\kappa_\phi}
\bigl(\ln\kappa_\phi+\ln\epsilon^{-1}\bigr)\bigr]$.
\end{corollary}

\begin{proof}
Since $1+n\sqrt{\kappa_\phi}\le2n\sqrt{\kappa_\phi}$ for $n\ge1$,
Theorem~\ref{thm:geometric-convergence} gives
\begin{equation}
  \frac{\chi_F-\chi_F^{(n)}}{\chi_F}
  \le16\kappa_\phi n^{2}e^{-2n\theta_\phi}
  =16\kappa_\phi\bigl(ne^{-n\theta_\phi/2}\bigr)^{2}e^{-n\theta_\phi}.
  \label{eq:error-factorization}
\end{equation}
The function $f(t)=te^{-t\theta_\phi/2}$ has
$f'(t)=e^{-t\theta_\phi/2}(1-t\theta_\phi/2)$ and hence a unique maximum
$f(2/\theta_\phi)=2/(e\theta_\phi)$, so that
\begin{equation}
  \frac{\chi_F-\chi_F^{(n)}}{\chi_F}
  \le\frac{64\kappa_\phi}{e^{2}\theta_\phi^{2}}\,e^{-n\theta_\phi}.
  \label{eq:error-before-kappa-bound}
\end{equation}
Applying $\operatorname{artanh}(u)\ge u$ on $[0,1)$ with
$u=\kappa_\phi^{-1/2}$ to Eq.~\eqref{eq:theta-rho} gives
$\theta_\phi\ge2/\sqrt{\kappa_\phi}$. Thus,
$\theta_\phi^{-2}\le\kappa_\phi/4$ and
\begin{equation}
  \frac{\chi_F-\chi_F^{(n)}}{\chi_F}
  \le\frac{16\kappa_\phi^{2}}{e^{2}}\,e^{-n\theta_\phi}.
  \label{eq:simplified-error-bound}
\end{equation}
The right-hand side is at most $\epsilon$ as soon as
$n\ge\theta_\phi^{-1}\ln\bigl(16\kappa_\phi^{2}/e^{2}\epsilon\bigr)$, and
the same lower bound on $\theta_\phi$ gives
$\theta_\phi^{-1}\le\sqrt{\kappa_\phi}/2$; taking the ceiling yields
Eq.~\eqref{eq:krylov-depth-budget}.

Finally, the order-$n$ Gram matrix has entries
$(\Gamma_n)_{ij}=M_{i+j+2}$ with $0\le i,j\le n-1$, whose largest index is
$2n$, attained at $i=j=n-1$; the moment vector involves only lower orders.
Hence $M_{2n}$ is the highest moment required, proving
Eq.~\eqref{eq:max-moment-order}.
\end{proof}
\section{Classical-Shadow Implementation}
\label{sec:classical_shadows}
Classical shadows provide a powerful method for estimating
many expectation values of an unknown quantum state without performing full
state tomography. Let $\rho$ be an unknown $N$-qubit state, and let $\mathcal{O}=\{O_i\}_{i=1}^{M}$ be a collection of observables whose expectation values $ o_i := \Tr(O_i\rho)$ we wish to estimate. The classical-shadow protocol produces an estimate $\hat{o}_i$. Interestingly, the choice of $\mathcal{O}$ is not fixed a priori, and the same set of measurement data can be reused for estimating many observables.

Within an additive error $\epsilon$ (with a failure probability of at most $\delta$), the required number of samples of $\rho$ scales as
\begin{align}
S&=
\mathcal{O}
\left(
\frac{\log ( M/\delta)}{\epsilon^2}
\max_i
\operatorname{Var}_{\mathcal{U}}
\!\left[
\hat{o}_i
\right]
\right) \nonumber \\
&=\mathcal{O}\left( \frac{\log(M/\delta)}{\epsilon^2} \max_{1 \le i \le M} \|O_i\|_{\text{shadow}}^2 \right) .
\label{eq:shadow_sample_complexity}
\end{align}
Therefore, the sample complexity depends only logarithmically on the number of observables $M$, while the dominant cost comes via the largest shadow norm among the observables of interest.

The sample complexity also depends strongly on the distribution $\mathcal{U}$ used in the randomized measurement. There are two popular choices. The first is a random $N$-qubit Clifford circuit, which yields a powerful estimator but requires $\mathcal{O}(N^2\log N)$ entangling gates per sampled circuit. The sample complexity is independent of the locality of $O$, but in the worst case, the cost scales with the full Hilbert-space dimension $d=2^N$. A more NISQ-friendly alternative is to use local Clifford measurements, or
equivalently, to measure each qubit independently in a random Pauli basis. In
this case, the protocol requires only single-qubit random gates and
computational-basis measurements. 

We will now outline the local Pauli classical-shadow protocol. For each experimental shot
$s=1,\ldots,S$, we independently choose, for every qubit $j$, $P_j^{(s)} \sim \mathrm{Unif}\{X_j,Y_j,Z_j\},$ and measure qubit $j$ in the eigenbasis of $P_j^{(s)}$ to get the outcome $\mu_j^{(s)}\in\{+1,-1\}.$ The classical data from one shot are therefore
$\{(P_j^{(s)},\mu_j^{(s)})\}_{j=1}^{N}$, and the $N$-qubit classical-shadow snapshot is
\begin{equation}
    \hat{\rho}^{(s)}
    =
    \bigotimes_{j=1}^{N}
    \frac{1}{2}
    \left(
        I+3\mu_j^{(s)}P_j^{(s)}
    \right).
    \label{eq:local_pauli_snapshot}
\end{equation}

By construction, this snapshot is unbiased, $\mathbb{E}\!\left[\hat{\rho}^{(s)}\right]=\rho$, and for each observable $O_i$, the single-shot estimator $  X_i^{(s)} := \Tr\!\left(O_i\hat{\rho}^{(s)}\right)$
satisfies $ \mathbb{E}\!\left[X_i^{(s)}\right] =\Tr(O_i\rho)=o_i .$ One can suppress the failure probability using a median-of-means estimator. For this, we partition the $S$ shots into $K$ groups $G_1,\ldots,G_K$ of size $m=S/K$, and define $ \bar{o}_{i}^{(g)}  =  \frac{1}{m}  \sum_{s\in G_g}  X_i^{(s)}$ for $g=1,\ldots,K .$ The final estimator is
\begin{equation}
    \hat{o}_i
    =
    \operatorname{median}
    \left\{
        \bar{o}_{i}^{(1)},
        \bar{o}_{i}^{(2)},
        \ldots,
        \bar{o}_{i}^{(K)}
    \right\}.
    \label{eq:median_of_means_shadow}
\end{equation}
With $S$ samples from Eq.~\eqref{eq:shadow_sample_complexity}, all $M$ expectation values are estimated to additive accuracy $\epsilon$ with failure probability at most $\delta$.  One downside of this approach could be that the shadow norm $\|O_i\|_{\text{shadow}}^2$ depends on the weight $w$ of the observables as  $\|O_i\|_{\text{shadow}}^2 \leq 4^w \|O_i\|_\infty^2 $. For local Pauli shadows, a Pauli string of weight $w$ has $\|P\|_{\mathrm{shadow}}^2 = 3^w,$ hence the protocol is efficient for estimating many low-weight observables.

When the observables are known in advance, the randomized protocol can be overkill, since many random measurement settings do not contribute. In this case, one may use a
derandomized classical-shadow protocol \cite{Huang_2021}, in which the Pauli measurement settings are chosen deterministically, or greedily, to cover the desired observables more
efficiently. The benefit is that the number of required measurements can be reduced substantially.

\begin{proof}[Proof of Theorem \ref{thm:certified-shadow-bound}]
  Let $\delta m_n:=\widehat{\widetilde m}_n-\widetilde m_n$,
$\delta\Gamma_n:=\widehat{\widetilde\Gamma}_n-\widetilde\Gamma_n$, and $ c_n:=\widetilde\Gamma_n^{-1}\widetilde m_n$. The moment bound of Eq.~\eqref{eq:sh-uniform-moment-error} gives
\begin{equation}
    \|\delta m_n\|_2\le\sqrt n\,\eta,
  \qquad
  \|\delta\Gamma_n\|_2\le\|\delta\Gamma_n\|_{\mathrm F}\le n\eta .
\end{equation}
Weyl's inequality gives the admissibility condition
\begin{align}
    \lambda_{\min}(\widetilde\Gamma_n)
    &\geq
    \lambda_{\min}
    \bigl(
        \widehat{\widetilde\Gamma}_n
    \bigr)
    -
    \|\delta\Gamma_n\|_2
    \geq
    \widehat\gamma_n-n\eta
    >
    0 .
\end{align}
Since $\widetilde m_n=\widetilde\Gamma_n c_n,$ and setting $r_n:=\delta m_n-\delta\Gamma_n c_n$, we get
\begin{align}
    \widehat{\widetilde m}_n
    &=
    \widetilde m_n+\delta m_n =
    \widehat{\widetilde\Gamma}_n c_n+r_n .
    \label{eq:normalized-residual-representation}
\end{align}

Substituting Eq.~\eqref{eq:normalized-residual-representation}
into $\widehat{\widetilde\chi}_F^{(n)}$ gives
\begin{align}
    \widehat{\widetilde\chi}_F^{(n)}
    &=
    \bigl(
        \widehat{\widetilde\Gamma}_n c_n+r_n
    \bigr)^\dagger
    \widehat{\widetilde\Gamma}_n^{-1}
    \bigl(
        \widehat{\widetilde\Gamma}_n c_n+r_n
    \bigr)
    \nonumber\\
    &=
    c_n^\dagger
    \widehat{\widetilde\Gamma}_n c_n
    +
    2\operatorname{Re}
    \bigl(
        c_n^\dagger r_n
    \bigr)
    +
    r_n^\dagger
    \widehat{\widetilde\Gamma}_n^{-1}
    r_n.
\end{align}
On the other hand,
\begin{equation}
    \widetilde\chi_F^{(n)}
    =
    \widetilde m_n^\dagger
    \widetilde\Gamma_n^{-1}
    \widetilde m_n
    =
    c_n^\dagger\widetilde\Gamma_n c_n.
\end{equation}
Using
$\widehat{\widetilde\Gamma}_n
=\widetilde\Gamma_n+\delta\Gamma_n$, we get
\begin{align}
    \widehat{\widetilde\chi}_F^{(n)}
    -
    \widetilde\chi_F^{(n)}
    &=
    2\operatorname{Re}
    \bigl(
        c_n^\dagger\delta m_n
    \bigr)
    -
    c_n^\dagger\delta\Gamma_n c_n+
    r_n^\dagger
    \widehat{\widetilde\Gamma}_n^{-1}
    r_n.
    \label{eq:normalized-exact-reconstruction-error}
\end{align}

The first two terms of Eq.~\eqref{eq:normalized-exact-reconstruction-error} satisfy
\begin{align}
    &2\operatorname{Re}
    \bigl(
        c_n^\dagger\delta m_n
    \bigr)
    -
    c_n^\dagger\delta\Gamma_n c_n
    \nonumber\\
    &\leq
    2\|c_n\|_2\|\delta m_n\|_2
    +
    \|c_n\|_2^2\|\delta\Gamma_n\|_2
    \nonumber\\
    &\leq
    \left(
        2\sqrt{n}\,\|c_n\|_2
        +
        n\|c_n\|_2^2
    \right)\eta.
    \label{eq:normalized-linear-error-bound}
\end{align}
Furthermore, $\left\|
        \widehat{\widetilde\Gamma}_n^{-1}
    \right\|_2
    =
    \frac{1}{\widehat\gamma_n},$ and the third term of Eq.~\eqref{eq:normalized-exact-reconstruction-error} satisfies
\begin{align}
    r_n^\dagger
    \widehat{\widetilde\Gamma}_n^{-1}
    r_n
    &\leq
    \frac{\|r_n\|_2^2}{\widehat\gamma_n}
    \nonumber\\
    &\leq
    \frac{
        \bigl(
            \|\delta m_n\|_2
            +
            \|\delta\Gamma_n\|_2\|c_n\|_2
        \bigr)^2
    }{
        \widehat\gamma_n
    }
    \nonumber\\
    &\leq
    \frac{
        \bigl(
            \sqrt{n}
            +
            n\|c_n\|_2
        \bigr)^2\eta^2
    }{
        \widehat\gamma_n
    }.
    \label{eq:normalized-quadratic-error-bound}
\end{align}

It remains to bound $\|c_n\|_2$ by data-dependent quantities, i.e., by the estimated normalized moments. Subtracting $\widetilde\Gamma_n c_n
    =
    \widetilde m_n$ from $\widehat{\widetilde\Gamma}_n\widehat c_n
    =
    \widehat{\widetilde m}_n$ gives
\begin{equation}
    \widehat{\widetilde\Gamma}_n
    \bigl(
        \widehat c_n-c_n
    \bigr)
    =
    \delta m_n-\delta\Gamma_n c_n
    =
    r_n.
    \label{eq:normalized-coefficient-residual}
\end{equation}
Using the triangle inequality,
\begin{align}
    \|c_n\|_2
    &\leq
    \|\widehat c_n\|_2
    +
    \|\widehat c_n-c_n\|_2
    \nonumber\\
    &\leq \|\widehat c_n\|_2 + 
    \frac{
        \|\delta m_n\|_2
        +
        \|\delta\Gamma_n\|_2\|c_n\|_2
    }{
        \widehat\gamma_n
    }\nonumber \\
    &\leq
    \|\widehat c_n\|_2
    +
    \frac{\sqrt{n}\,\eta}{\widehat\gamma_n}
    +
    \frac{n\eta}{\widehat\gamma_n}
    \|c_n\|_2 \nonumber \\
    &\leq
    \frac{
        \|\widehat c_n\|_2
        +
        \sqrt{n}\,\eta/\widehat\gamma_n
    }{
        1-n\eta/\widehat\gamma_n
    }
    =:
    \overline c_n.
\label{eq:normalized-exact-coefficient-bound}
\end{align}

The right-hand sides of Eqs.~\eqref{eq:normalized-linear-error-bound} and
\eqref{eq:normalized-quadratic-error-bound} are monotonically
increasing functions of $\|c_n\|_2$. Substituting
Eq.~\eqref{eq:normalized-exact-coefficient-bound} therefore gives
\begin{align}
    2\operatorname{Re}
    \bigl(
        c_n^\dagger\delta m_n
    \bigr)
    -
    c_n^\dagger\delta\Gamma_n c_n
    &\leq
     L_n\eta,
    \\
    r_n^\dagger
    \widehat{\widetilde\Gamma}_n^{-1}
    r_n
    &\leq
    R_n \eta^2.
\end{align}
Combining these estimates with
Eq.~\eqref{eq:normalized-exact-reconstruction-error} gives
\begin{equation}
    \widetilde\chi_F^{(n)}
    \geq
    \widehat{\widetilde\chi}_F^{(n)}
    -
     L_n\eta
    -
     R_n \eta^2.
\end{equation}

Finally, the normalized Krylov hierarchy satisfies $\widetilde\chi_F
    \geq
    \widetilde\chi_F^{(n)}$, and $\chi_F
    =
    \Omega^{-2}\widetilde\chi_F$. Combining these two facts with the previous display, we obtain the claim
\begin{equation}
    \chi_F
    \geq
    \Omega^{-2}
    \left[
        \widehat{\widetilde\chi}_F^{(n)}
        -
        L_n\eta
        -
         R_n \eta^2
    \right] .
\end{equation}
\end{proof}

\end{document}